\documentclass{sig-alternate-05-2015}

\RequirePackage[colorlinks,citecolor=blue,urlcolor=blue]{hyperref}
\usepackage{algorithm,algpseudocode}
\RequirePackage{amsmath}
\RequirePackage[numbers]{natbib}
\usepackage{amsfonts,amssymb,graphicx,nicefrac,mathtools}
\usepackage{enumerate, dsfont}
\usepackage{paralist}

\newtheorem{remark}{Remark}
\newtheorem{lemma}{Lemma}

\input{Definitions}

\begin{document}

\setcopyright{acmcopyright}
%
\doi{}

\isbn{}

\conferenceinfo{KDD '16}{August 13--17, 2016, San Francisco, CA, USA}

\acmPrice{\$15.00}

%

\title{Large scale multi-objective optimization: Theoretical and practical challenges}
%
%
%
%
%

\numberofauthors{3} 
%
\author{
%
%
\alignauthor
Kinjal Basu\\
       \affaddr{Department of Statistics}\\
       \affaddr{Stanford University}\\
       \affaddr{Stanford, CA USA}\\
       \email{kinjal@stanford.edu}
\alignauthor Ankan Saha\\
       \affaddr{LinkedIn Corporation}\\
       \affaddr{Mountain View, CA}\\
       \email{asaha@linkedin.com}
\alignauthor Shaunak Chatterjee \\
       \affaddr{LinkedIn Corporation}\\
       \affaddr{Mountain View, CA}\\
       \email{shchatterjee@linkedin.com}
}

\maketitle
\begin{abstract}
Multi-objective optimization (MOO) is a well-studied problem for
several important recommendation problems. While multiple approaches
have been proposed, in this work, we focus on using constrained
optimization formulations (e.g., quadratic and linear programs) to
formulate and solve MOO problems. This approach can be used to pick
desired operating points on the trade-off curve between multiple
objectives. It also works well for internet applications which serve
large volumes of online traffic, by working with Lagrangian duality
formulation to connect dual solutions (computed offline) with the
primal solutions (computed online).

We identify some key limitations of this approach -- namely the
inability to handle user and item level constraints, scalability
considerations and variance of dual estimates introduced by sampling
processes. We propose solutions for each of the problems and
demonstrate how through these solutions we significantly advance the
state-of-the-art in this realm. Our proposed methods can exactly
handle user and item (and other such local) constraints, achieve a
$100\times$ scalability boost over existing packages in R and reduce
variance of dual estimates by two orders of magnitude.
%
%

\end{abstract}

%
%


%
%

%
%


\section{Introduction}
\label{sec:intro}
Today, most internet applications and products use data to optimize
the user experience. With every passing year, more and more such
applications are coming into existence, and the scales of the existing
applications are increasing. The businesses operating these
applications are also growing and so are the users' expectations from
these applications. A natural manifestation of this multi-dimensional growth is through the introduction of multiple objectives (e.g.,
company's business objectives and user engagement). While a search
engine could initially focus on maximizing the click-through rate
(CTR) of its first few results, as the business grows, monetization
and user retention become additional objectives that need to be
ensured. This necessitates a framework which allows for efficiently
navigating the trade-off between such important objectives.

\begin{table*}[t]
\caption{Existing constrained optimization solution -- setup}
\label{tab:existing_setup}
 \begin{center}
 \begin{tabular}{|| l || l | l ||} 
 \hline
 \multicolumn{1}{||c||}{Stage 1} & \multicolumn{2}{|c||}{Stage 2} \\
 \hline
  &  \textbf{Offline application} & \textbf{Online application} \\
 Step 1.1: Sample a set of users, as large &  & \\
as can be handled by your QP solver   & (Batch processing) Iterate over users.  & As each user $u$ visits, use  \\
 	      	& For each user $u$, use Algorithm \ref{dual1}\protect\footnotemark
  & 
with inputs $\vec{p_u}$ (scored  \\
 						&  with inputs $\vec{p_u}$ (scored offline) and & online) and  $\vec{\mu}$ (from stage 1) to get \\
 Step 1.2: Solve the QP for the sampled  &   $\vec{\mu}$ (from stage 1) to get $\vec{x_u}$ & $\vec{x_u}$   \\
 users. Get the dual variables $\vec{\mu}$.    &    &  \\
    &  & \\

 \hline
\end{tabular}
\end{center}
\end{table*}
\addtocounter{footnote}{0} \footnotetext{Algorithm \ref{dual1} in this
  paper is presented for a slightly more advanced case. However, the
  actual algorithm that should be used is the dual to primal
  conversion (Algorithm 1) in \cite{agarwal2012personalized}. This is
  true for all references of Algorithm \ref{dual1} in the two tables
  above. We avoid copying it here due to space constraints.}

\begin{table*}[t]
\caption{New constrained optimization solution -- setup}
\label{tab:new_setup}
 \begin{center}
 \begin{tabular}{|| l || l | l ||} 
 \hline
 \multicolumn{1}{||c||}{Stage 1} & \multicolumn{2}{|c||}{Stage 2} \\
 \hline
  &  \textbf{Offline application} & \textbf{Online application} \\
  
 Step 1.1: Sample a set of users, as large &  & \\
 
 as can be handled by your QP solver.     & (Batch processing) Iterate over users. & \textbf{\textcolor{blue}{Step 2.1: (Offline) Iterate over}}    \\
 
 \textbf{\textcolor{blue}{Now we can solve 100x larger QPs}} & For each user $u$, use \textbf{\textcolor{blue}{Algorithm \ref{algo:dual2}}} & \textbf{\textcolor{blue}{the users. For each user $u$,}} \\
 
  \textbf{\textcolor{blue}{ with ADMM -- operator splitting}} &   with inputs $\vec{p_u}$ (scored offline) and & \textbf{\textcolor{blue}{use Algorithm \ref{algo:dual2} with inputs}}  \\
  
   &    $\vec{\mu}$ (from stage 1) to get $\vec{x_u}$ & \textbf{\textcolor{blue}{$\vec{p_u}$ (scored offline) and $\vec{\mu}$ }} \\
   
Step 1.2: Solve the QP for the sampled &   & \textbf{\textcolor{blue}{(from stage 1) to get $\vec{\mu_u}$}} \\

 users. Get the dual variables $\vec{mu}$ &  &  \\
 
 \textbf{\textcolor{blue}{Variance of estimated $\vec{\mu}$ is greatly}} & & Step 2.2: As each user $u$ visits, \\
 
  \textbf{\textcolor{blue}{reduced with variance reduction}}  &  & use Algorithm \ref{dual1} with inputs $\vec{p_u}$ \\
  
  \textbf{\textcolor{blue}{techniques}}     &  &  (scored online), $\vec{\mu}$ (from stage 1) \\
  
	&  &  and \textbf{\textcolor{blue}{$\vec{\mu_u}$ (obtained from step 2.1)}} \\
 
 & & to get $\vec{x_u}$ \\
 \hline
\end{tabular}
\end{center}
\end{table*}

Since we are almost into the third decade of optimizing internet
products with data, it is not surprising that several approaches have
been proposed to address this problem \cite{adomavicius2011multi,
  agarwal2015constrained, agarwal2012personalized, deb2014multi,
  konak2006multi, marler2004survey, rodriguez2012multiple} --- both in
theory and in practice. The common thread in these works is trying to
combine several objectives or criteria. In this paper, we will focus
on the constrained optimization approach \cite{agarwal2015constrained,
  agarwal2012personalized} since this provides a lot of flexibility on
the problem formulation. The relative importance of each objective need not
be pre-specified, instead the desired range of value of each
objective is specified, and the relative weights (to satisfy those constraints) are learnt. This
added flexibility in specification, combined with the scalability of
the approach, has made this a feasible solution for industrial
applications.

Industrial systems employing machine learning algorithms can be
broadly categorized into two classes --- offline systems and online
systems, where online is defined as triggered by a user visit. In
offline systems, the entire computation is done offline and the
results (e.g., best article recommendations for a user) are
pre-computed and used to serve a user when they visit. Another example
of such a system is an email delivery system where machine learning
algorithms can be used to determine both the content and
appropriateness of an email for a user. In an online system, on the
other hand, the candidate items are scored when the user visits. These
systems naturally have strict computation constraints, but also
provide better results for the same application than an offline
approach, since there are more recent signals to use.

When constraints are used to specify multi-objective optimization problems, then a two-stage approach is adopted for both offline and online applications. The two-stage approach is necessary since the problem size is too large to be solved as a whole. In Stage $1$ (which happens offline), we sample a set of users from the entire population. We then solve the constrained optimization problem for this sample, and obtain optimal duals corresponding to each constraint. This constrained optimization problem has to be a quadratic program, since linear programs cannot facilitate the dual to primal conversion which is required next. In the second stage (which happens offline for offline applications, and online for online ones), the dual estimates from Stage $1$ are used to convert each user's (or user visit's) parameters into the primal serving scheme. Details of this method are provided in \cite{agarwal2015constrained, agarwal2012personalized}. The setup is also depicted in details in Table \ref{tab:existing_setup}.


Over the past couple of years, we have been extensively building and deploying such
constrained optimization solutions for both offline and online systems for
various applications within Linkedin. From our experience, we have
identified certain practical and theoretical challenges which were
previously not addressed, but have proven crucial to the success of
those endeavors. This work identifies these challenges, and presents
solutions for them.

The first challenge is a theoretical one. The existing solutions
cannot handle any non-global constraints. However, we show a
mathematical way to handle user level constraints (i.e., local
constraints). Having solved this problem, we explore various possible
combinations of constraints spanning the entire continuum of global to
local, and formulate solution trade-offs in terms of computation time
and accuracy. We propose a mathematical formulation that can help in
deciding which parts of the problem to be solved offline, and which
parts are solved online to obtain the right trade-off between
computation time and accuracy.

The second challenge was one of scaling up a quadratic programming
(QP) solver. While this is a well-studied theoretical problem, there
isn't any documented work on deploying such systems in real-world
applications. We explore multiple methods, and share our experiences
along with intuitions on what works best and how they scale.

The third challenge we address is related to sampling. The dual
estimates obtained from Stage $1$ can have high variance, as an artifact of the sampling process. We employ
some variance reduction methods and show that they reduce the
variance of the dual estimates, and hence improve the accuracy of
the final (i.e., Stage $2$) solution which depends on those dual estimates. With
these advances, the new state of the ecosystem is briefed in Table
\ref{tab:new_setup} (all the advances are marked in bold and blue).

The rest of the paper is structured as follows. Section
\ref{sec:problem_def_global} formulates the problem and describes the
existing solution in detail. Section \ref{sec:problem_def_local}
introduces the importance of local constraints, highlights the
limitations of the existing solution in handling these and presents a
mathematically exact solution to solve the problem. We present a
graphical representation of the constraint set in Section
\ref{sec:graph_cal} and introduce a criteria to pick
which constraints to re-solve during serving time. Section
\ref{sec:scaling_QP} discusses our experience in using ADMM algorithms
to scale up a single-machine QP solver and some variance reduction
methods. Finally, we present some experiments to validate our methods
and quantify their impact in Section \ref{sec:results}.

\begin{table}[!t]
  \centering
  \caption{ \label{tab:Symbols} List of Notations used in the paper}
  \begin{tabular}{|l|r|}
    \hline
    \bf{Notation} &
    \bf{Meaning} \\    
    \hline
    $I$ &
    Set of all items to be shown \\
    \hline
    $U$ & 
    Set of all Users\\
    \hline
    $x_{ui}$ &
    Probability of showing item $i$ to user $u$.\\
    \hline
    $p_{ui}$ &
    Pr(User $u$ engaged with item $i$|$x_{ui}=1$) \\
    \hline
    $r_{ui}$ &
    Pr(User $u$ disliking/complaining\\
    & about item $i$|$x_{ui}=1$) \\
    \hline
    $\sum_{u,i}x_{ui}p_{ui}$ &
    Expected number of clicks\\
    \hline
    $\one_N,\zero_N$ & 
    Vector of 1's (0's) of dimension $N$\\
    \hline
    $\Ccal$ & 
    $[0,1]^{MN}$ \\
    \hline
  \end{tabular}
\end{table}

\vspace{-1em}
\section{The MOO Problem}
\label{sec:problem_def_global}
In this paper we are mainly interested in the recommendation problem
of showing items to users such that will maximize their engagement while
ensuring that the potential negative flags or complaints (for
simplicity, we will just refer to them as complaints hereafter) are
contained within a limit.

We begin by introducing some notation, which we will use throughout
the length of the paper. Lower bold case letters (e.g.,
$\pb, \xb$) denote vectors, while upper bold case letters (e.g.,
$\Ab$) denote matrices or linear operators. We use $p_{i,j,k}$ to
denote the appropriate index of $\pb$ (similarly for matrices) and
$\xb^T\zb = \sum_i x_i z_i$ to denote the Euclidean dot product between
$\xb$ and $\zb$. The relation $\leq$ when applied to a vector implies
element wise inequalities. Also $ \Pi_{\Ccal} (\cdot) $ denotes the
projection operator onto the set $\Ccal$ in terms of the $L_2$ norm.

Let a user be denoted as $u$ and the set of all users $U$, such that
$|U| = N$. Similarly, let an item be denoted as $i$ and
the set of items as $I$, such that $|I| = M$. The target
serving plan (which we seek to find) can be represented as $x_{ui}$,
which is the probability of showing item $i$ to user $u$. Let $p_{ui}$
denote the probability of the user $u$ engaging (by acting on it or
clicking on it) with item $i$, and $r_{ui}$ denote the probability of
user $u$ disliking item $i$ (e.g., by flagging it or complaining about
it) conditioned on the fact that the user was shown the item $i$.  A
detailed list of all symbols and their meanings is provided in Table
\ref{tab:Symbols}.
\subsection{Problem Formulation}
The aforementioned optimization problem can be written as a linear
program. Using the above notations we can write it as
\begin{equation*}
\begin{aligned}
& \underset{x}{\text{Maximize}} &  &\sum_{u,i}x_{ui}p_{ui}\\
& \text{subject to}  & & \sum_{u} \sum_{i} x_{ui}r_{ui} \leq R\\
& &  & 0 \leq x_{ui} \leq 1, \; \forall u,i,   \qquad \sum_{i}x_{ui} = 1,\; \forall u.
\end{aligned}
\end{equation*}
The last two constraints come from the fact that $x_{ui}$ is a
probability and there will be some item shown to every user. If there
is an existing serving plan that we do not want to deviate too much
from (e.g., showing the most engaging item to a user), then this can
be represented in the following way (similar
to \cite{agarwal2012personalized}). Let $q_{ui}$ represent the
existing serving plan (in the example, $q_{ui^*}$ = 1, where $i^* =
\argmax_{i} p_{ui}$). Then, we can limit the deviation of $\xb$ from $\qb$
with the following quadratic program (QP):
\begin{equation}
\label{eq:qp1}
\begin{aligned}
& \underset{x}{\text{Maximize}} & & \sum_{u,i}x_{ui}p_{ui}
- \frac{\gamma}{2} \sum_{u,i} (x_{ui} - q_{ui})^2\\ & \text{subject
to}  & & \sum_{u} \sum_{i} x_{ui}r_{ui} \leq R\\  &  & & \sum_{i}x_{ui} =
1,\; \forall u  \\& & & 0 \leq x_{ui} \leq 1, \; \forall u,i,
\end{aligned}
\end{equation}
where $\gamma$ controls the relative importance of engagement
maximization and deviation of $\xb$ from $\qb$. Note that the objective
in \eqref{eq:qp1} is concave, which would be equivalent to minimizing
the negative of the expression. The QP formulation also facilitates
some additional algorithms when we are straddling the dual and primal
spaces to work out solutions. 

This formulation can be easily extended to add objectives or
constraints to a large set of users or items. For instance, if we want
to put an additional constraint on the number of complaints from
English-speaking users (denoted by $U_{En}$), then the modified QP
would look like:
\begin{equation*}
\begin{aligned}
& \underset{x}{\text{Minimize}} & & -\sum_{u,i}x_{ui}p_{ui}
+ \frac{\gamma}{2} \sum_{u,i} (x_{ui} - q_{ui})^2\\ 
& \text{subject to} & & \sum_{u} \sum_{i} x_{ui}r_{ui} \leq R\\ 
& & & \sum_{u \in U_{En}} \sum_{i} x_{ui}r_{ui} \leq R_{En}\\ 
& & & \sum_{i}x_{ui} = 1,\; \forall u \\ 
& & & 0 \leq x_{ui} \leq 1, \; \forall u,i.
\end{aligned}
\end{equation*}

\subsection{Solving the MOO problem}
\label{sec:solving_moo}
We shall modify the technique in \cite{agarwal2012personalized} to
solve for the primal $x_{ui}$ using the dual of the problem. For
brevity, we work with problem \eqref{eq:qp1} as all the derivations
can be trivially extended to the other formulations of the problem as
well. We define the following notation first.

Let $\ab \in \RR^{MN}$ be given by $a_{ui} = p_{ui} + \gamma q_{ui}$, 
$\rb$ be the vectorized form of $r_{ui}$ and $\xb_u = (x_{u1}\hdots
x_{uM})^T$. Then the problem \eqref{eq:qp1} can be transformed into
\begin{equation}
\label{eq:origqp2}
\begin{aligned}
& \underset{\xb}{\text{Minimize}} & & - \xb^T\ab + \frac{\gamma}{2} \xb^T\xb \\
& \text{subject to} & & \xb^T\rb \leq R\\
& & & \xb_u^T\one = 1 \qquad \forall u \\
& & & \zero \leq \xb \leq \one.
\end{aligned}
\end{equation} 
Note that the second constraint can also be written in the format 
\[
 \Mb^T\xb = \one_N,
\]
where $\one_N \in \RR^N$ is the vector of all $1$'s in $N$ dimension and
$\Mb \in \RR^{MN \times N}$ is given by $M = \Ib_{N} \otimes \one_M$ and $\otimes$ denotes the Kronecker product.
Note that, we can write the Lagrangian as follows. Let
$\mu \in \RR, \nub \in \RR^N$, $\xib \in \RR^{MN}$ and
$\eta \in \RR^{MN}$ be the dual variables corresponding to the above
problem \eqref{eq:origqp2}. Thus, we have
\begin{align*}
\Lcal(\xb, \mu, \nub, \xib, \etab) &= -\xb^T\ab + \frac{\gamma}{2} \xb^T\xb + \mu(\xb^T\rb - R) \\
  & + \nub^T(\Mb^T\xb -\one_N) -\xib^T\xb + \etab^T(\xb-\one).
\end{align*}
Minimizing this with respect to $\xb$, and writing $\tilde{\Ab} = [-\rb :
-\Mb : \Ib : -\Ib]$ and $\zb = (\mu, \nub, \xib, \etab)$ we have
\begin{align*}
\xb^* &= \frac{1}{\gamma} \left(\ab - \mu\rb  - \Mb\nub + \xib - \etab \right)\\
&= \frac{1}{\gamma}\rbr{\ab + \tilde{\Ab}\zb}.
\end{align*}
Letting $\sbb = \rbr{-R, - \one_N, \zero_{NM} , - \one_{NM}}$ and plugging $\xb^*$ back into the
Lagrangian we get,
\begin{align*}
L(\xb^*, \mu, \nub, \xib, \etab) &= -\mu R - \nub^T\one_N - \etab^T\one_{NM}\\
& - \frac{1}{2\gamma}\rbr{\ab + \tilde{\Ab}\zb}^T\rbr{\ab + \tilde{\Ab}\zb}\\
&= \sbb^T\zb - \frac{1}{2\gamma}\rbr{\ab + \tilde{\Ab}\zb}^T\rbr{\ab + \tilde{\Ab}\zb}.
\end{align*}
Duplicating the equality constraint as a positive as well as a
negative constraint (using $\nub_+$ and $\nub_-$) and having
$\tilde{\zb} = (\mu, \nub_+, \nub_-, \xib, \etab),\; \tilde{\sbb}
= \rbr{-R : - \one_N : \one_N : \zero_{NM} : - \one_{NM}}$ and
$\hat{\Ab} = [-\rb : -\Mb : \Mb : \Ib : -\Ib]$, we can write our dual
problem as
\begin{equation*}
\begin{aligned}
& \underset{\tilde{\zb}}{\text{Minimize}} & & -\tilde{\sbb}^T\tilde{\zb} + \frac{1}{2\gamma}(\ab
+ \hat{\Ab}\tilde{\zb})^T(\tilde{\zb} + \hat{\Ab}\tilde{\zb})\\ 
& \text{subject to} & &\tilde{\zb} \geq 0.
\end{aligned}
\end{equation*}
Writing $\Bb
= \hat{\Ab}^T\hat{\Ab}/ \gamma,\; \tilde{\pb} = \tilde{\sbb} - \hat{\Ab}^T\ab/\gamma$,
we can re-write the above problem in the most basic form,
\begin{equation}
\label{eq:dual2}
\begin{aligned}
& \underset{\zb}{\text{Minimize}} & & \frac{1}{2}\zb^T\Bb\zb  - \tilde{\pb}^T\zb\\
& \text{subject to} & & \zb \geq 0.
\end{aligned}
\end{equation}
This problem can now be solved by any convex optimization algorithm. Large scale instances of the the above problem can be solved efficiently by the Operator Splitting algorithm \cite{parikh2014block}. For more details see Section \ref{sec:scaleup}.

\section{User-level Constraints}
\label{sec:problem_def_local}
The formulation described above provides an optimization problem that
can be solved efficiently when the constraints are global in nature
\ie\ applicable to all users simultaneously. However, we would often
want to apply different kinds of constraints for specific users or
items belonging to particular groups. Some examples of this include
applying revenue threshold constraints for sponsored items on the feed
and having different complaint rates for a certain subset of heavy
users. To the best of our knowledge, the introduction of such
user-level (local) constraints has not been explored in detail in
previous literature.

Note that the number of constraints might significantly increase with
the introduction of user-level or item-level constraints which makes
scaling up the optimization problem to large scale data a significant
challenge.  In this section, we provide a scheme to solve this problem
in two steps. In the first step, we get the global dual corresponding
to the first constraint using the technique in
Section \ref{sec:solving_moo} and using this dual we solve for $\xb$
to get the user-level local coefficients. We begin with a specific
example to explain our procedure and then we generalize to it to a
much larger class of constraints.

\subsection{Solving the primal using dual solution}
Consider a modification of the original primal
problem \eqref{eq:origqp2} so that the user $u$ can see a set of items
$i$ (for example messages) under the constraint that no user will be
shown more than $K_u$ messages. Using $\kb = \cbr{K_u}_{u=1}^N$, this
constraint can be written as $\Mb^T \xb \leq \kb$. Note that the sum
to $1$ constraint no longer holds in this case. The problem can now be
written in an analogous way by pushing some of the constraints into
the objective.
\begin{equation}
\label{eq:finalqp}
\begin{aligned}
& \underset{\xb}{\text{Minimize}} & & -\ab^T\xb + \frac{\gamma}{2} \xb^T\xb +
\II_\Ccal(\xb)\\ 
& \text{subject to} & & \xb^T\rb \leq R\\ 
& & & \Mb^T\xb \leq \kb, 
\end{aligned}
\end{equation}
where $\Ccal = [0,1]^{MN}$ and $\II_{\Ccal}(\xb) = 0$ if
$\xb \in \Ccal$ and infinity otherwise. We can write the Lagrangian
function of the problem in
\eqref{eq:finalqp} as
\begin{align*}
L(\xb, \mu, \nub) &= -\ab^T\xb + \frac{\gamma}{2} \xb^T\xb +
\II_{\Ccal}(\xb) + \mu(\xb^T\rb - R)\\
&\qquad + \nub^T(\Mb^T\xb - \kb). 
\end{align*}
After some algebra, it is easy to see at the optimal $\xb$ can be
written as
\begin{align*}
\xb^* = \Pi_{\Ccal} \left( \frac{\ab - \mu\rb  - \Mb\nub}{\gamma} \right),
\end{align*}
where $ \Pi_{\Ccal} (\cdot) $ stands for the projection into $\Ccal$.
Going into the specific user and item level, we can write the above
equation as
\begin{align}
\label{eq:xui}
x^*_{ui} = \Pi_{[0,1]} \left(\frac{ c_{ui} - \nu_u }{\gamma}\right),
\end{align}
where $c_{ui} = a_{ui} - \mu r_{ui}$ and $\nu_u$ is the coefficient of
$\nub$ corresponding to user $u$ which comes out by multiplication
with $\Mb$.


 In \cite{agarwal2012personalized}, a trick is used to eliminate the
 user level variables such that the serving plan of the next epoch can
 be calculated using just the global dual variable $\mu$ from the
 previous epoch. However
 they would be limited to handling only specific types of
 constraints. The trick we used by removing the box constraints and
 introducing the indicator variable removes extra dependency on other
 dual variables and makes it easier to recover $x_{ui}$ from $\mu$. We
 begin with a short lemma.
 
\begin{lemma}
If $c_{ui_1} \geq c_{ui_2}$, then $x^*_{ui_1} \geq x^*_{ui_2}$.
\end{lemma}
\begin{proof}
Note that since $\nu_u$ is common between the two, it is easy to see
that $c_{ui_1} - \nu_u \geq c_{ui_2} - \nu_u$. The result follows by
observing that $\Pi_{[0,1]}(\cdot)$ maintains the order.
\end{proof}

Thus, if we sort $c_{ui_1} \geq c_{ui_2} \geq \ldots \geq c_{ui_M}$,
then $x_{ui_1} \geq x_{ui_2} \geq \ldots \geq x_{ui_M}$. This implies,
there exists $t_1$ and $t_2$ with $0 \leq t_1 < t_2 \leq M$, such that
\begin{align}
\label{x_final}
x^*_{ui} = \begin{cases}
    1, & \text{for $i = i_1 , \ldots, i_{t_1}$}.\\
    (c_{ui} - \nu_u)/\gamma, & \text{for $i = i_{t_1 + 1}, \ldots, i_{t_2}$}\\
    0, &  \text{for $i = i_{t_2 + 1}, \ldots, i_{M}$}.
  \end{cases}
\end{align}
Note further that from complementary slackness for optimality
conditions we have $\nu_u > 0$ if and only if $\sum_{i} x_{ui} =
K_u$. Thus using this equation we can solve for $\nu_u$. In fact it is
easy to see that
\begin{align}
\label{nu_eqn}
\nu_u = \frac{\gamma(t_1 - K_u) + \sum_{j = t_1 + 1}^{t_2}c_{ui_j}}{t_2 - t_1}.
\end{align}
We can now formally write the algorithm as follows. 
\begin{algorithm}[!h]
\caption{Optimal Primal from Dual}\label{dual1}
\begin{algorithmic}[1]
\State \text{Input : Optimal dual solution $\mu$ and an incoming user $u$}
\State \text{Output : Primal serving plan $\{x_{ui}\}$}
\State \text{Compute $c_{ui}$ for all $i \in I$} 
\State Sort $c_{ui}$ in decreasing order as $c_{ui_1}, c_{ui_2}, \ldots, c_{ui_M}$.
\State Generate all combinations of $(t_1, t_2)$ combinations such that $0 \leq t_1 < t_2 \leq M$.
\For {each pair $(t_1, t_2)$}
\State Compute $\nu_u$ by equation \eqref{nu_eqn}
\If {all the following conditions hold} {break from loop}
\State $(c_{ui_{t_1}} - \nu_u)/\gamma > 1$
\State $(c_{ui_{t_1 + 1}} - \nu_u)/ \gamma < 1$ and $(c_{ui_{t_2}} - \nu_u) /\gamma > 0$
\State $(c_{ui_{t_2 + 1}} - \nu_u)/ \gamma < 0$
\EndIf
\EndFor
\State Compute $x_{ui}$ using \eqref{x_final}
\State \Return $\{x_{ui}\}$.
\end{algorithmic}
\end{algorithm}
\subsection{Generalizations of local constraint}
Let us now consider a case of more general local constraints. In
\eqref{eq:origqp2} if we want to ensure that each user $u$ is shown at
most $K_u$ items of a particular type (belonging to a set $\Jcal_u$),
then an additional set of constraints (one for each user) will be
added:
\begin{equation}
\label{eq:qp3}
\begin{aligned}
& \underset{\xb}{\text{Minimize}}  & & - \xb^T\ab + \frac{\gamma}{2} \xb^T\xb \\
& \text{subject to} & & \xb^T\rb \leq R\\
& & & \sum_{i \in \Jcal_u} x_{ui} \leq K_u, \qquad \forall u \\ 
& & & \zero \leq \xb \leq \one
\end{aligned}
\end{equation}
These local constraints are given by $\Pb\xb \leq \kb$ where $\Pb \in
\RR^{N \times MN}$. However, we might be interested in more
generalized local constraints such as for $\Jcal_{u0} \subset \Jcal_u$, we
might want $\sum_{i \in \Jcal_{u0}} x_{ui} \geq n_1^u$. 
Adding such constraint to the optimization problem introduces a new
set of dual variables which we need to be eliminated to apply the above
algorithm. This may not be possible in all cases since the sign of the
combination of dual variables may not be known and hence the ordering
of $c_{ui}$ which translates to the ordering in $x_{ui}$ breaks
down. Here we present an efficient algorithm to work around
this problem. Without loss of generality, we consider the local
constraint at user level $u$ by considering $\xb_{u} \in \Kcal_u$,
where the convex region $\Kcal_u$ is defined by any sort of linear
constraints. Now we can write the optimization problem as
\begin{equation*}
\begin{aligned}
& \underset{\xb}{\text{Minimize}} & & -\ab^T\xb  + \frac{\gamma}{2} \xb^T\xb\\
& \text{subject to} & & \xb^T\rb \leq R \\
& & & \xb_u \in \Kcal_u, \; \forall u.
\end{aligned}
\end{equation*}
We transform this by introducing indicator variables as before.
\begin{equation*}
\begin{aligned}
& \underset{\xb}{\text{Minimize}} & & -\ab^T\xb + \frac{\gamma}{2} \xb^T\xb
+ \sum_{u}\II_{K_u}(x_u)\\ 
& \text{subject to} & & \rb^T\xb \leq R.
\end{aligned}
\end{equation*}
Similar to \eqref{eq:xui}, using the Lagrangian, it is easy to see
that
\begin{align*}
\xb_u = \Pi_{\Kcal_u}\left( \frac{(\ab - \mu \rb)_u}{\gamma} \right).
\end{align*}
This follows because we the can write the entire convex domain of $x$
as $\Kcal_1 \times \Kcal_2 \times \ldots \times \Kcal_n$ and each
$\xb_u$ is the projection into the corresponding $\Kcal_u$. This is
the most important step in our case, because unless we are able to
split the domain into user size, we cannot apply this
decomposition. 

Now once we know $\mu$ and the user $u$ we can calculate $\cbb_u =
\frac{(\ab - \mu\rb)_u}{\gamma}$. Then we only need to project
$\cbb_u$ into $K_u$. For most practical purposes, the number of
candidate items considered for ranking (after initial filtering based
on recency, quality and other basic filters) is at a much smaller
scale. This leads us to the following  
QP optimization problem of reasonable scale,
\begin{equation}
\label{eq:genqp_dual}
\begin{aligned}
& \underset{\xb_u}{\text{Minimize}} &  &\norm{\xb_u -   \cbb_u}_2^2\\
& \text{subject to} & &  \xb_u \in \Kcal_u,
\end{aligned}
\end{equation}
which can be solved by any QP solver. The detailed generalized
conversion technique is given next in Algorithm \ref{algo:dual2}.
\begin{algorithm}
\caption{Generalized conversion algorithm}\label{algo:dual2}
\begin{algorithmic}[1]
\State \text{Input : Optimal dual solution $\mu$ and an incoming user $u$}
\State \text{Output : Primal serving plan $\{x_{um}\}$}
\State \text{Compute $\cbb_{u} = (\ab  - \mu \rb)_u / \gamma$} 
\State Solve the problem \eqref{eq:genqp_dual} to get $\xb_u$.
\State \Return $\xb_u$.
\end{algorithmic}
\end{algorithm}

\vspace{-1em}
\section{Generalized Graph Structure}
\label{sec:graph_cal}
In this section we provide a mathematical framework based on directed
acyclic graphs (DAG) which can help in deciding what kind of local
constraints should be solved in an online setting. We begin by
describing how we obtain a graph from the set of constraints of a
general multi-objective optimization problem such
as \eqref{eq:qp3}. We then elaborate an algorithm for getting an
optimal split into offline and online problems.
\subsection{Graph Construction}
\label{sec:graph_const}
Let us begin with a few notation. Let $\xb = (x_1, \ldots, x_n)$ be
the optimization variable in $n$-dimensions. Furthermore, let $I_n
= \{1, \ldots, n\}$ and $S_j^k$ denote subsets of $I_n$ such that
$|S_j^k| = k$ for $j = 1, \ldots, s_{k}$ for some $s_{k} \geq 1$ and
$k \in \{1, \ldots, n\}$. Note that $S_1^n = I_n$. For ease of
notation, let $\pb, \cbb$ denotes some parameters in the optimization
problem, which may change in each constraint. However, for notational
simplicity we use $\pb, \cbb$. Now consider the following optimization
problem,
\begin{equation}
\label{eq:dag}
\begin{aligned}
& \underset{\xb }{\text{Minimize}} & &f(\xb)\\
& \text{subject to} & & \sum_{i \in I_n} x_i p_ i \geq c_1^n \\
& & &\sum_{i \in S_j^{n-1}} x_i p_i \geq c_j^{n-1} \text{ for } j = 1, \ldots, s_{n-1}\\
& & &\ldots\\
& & & x_i \geq c_i^{1} \text{ for } i = 1, \ldots, s_1.
\end{aligned}
\end{equation}
Here we have constraints on a subset of variables. It is easy to see
that the problem \eqref{eq:qp3} is a specific case of \eqref{eq:dag}.

Now we construct the DAG, $G = (V,E)$ as follows. Each set $S_j^k$
corresponds to a node $n_j^k$. For $\ell < k$, there is a directed
edge from $n_j^{k}$ to $n_{j'}^{\ell}$ if $S_{j'}^{\ell} \subset
S_{j}^{k}$ and there does not exist a $k'$ with $\ell < k' < k$ such
that $S_{j'}^{\ell} \subset S_{\tilde{j}}^{k'}$ for some
$\tilde{j}$. There also exists an edge from $n_j^{k}$ to
$n_{j'}^{\ell}$ if, there exists an $x \in S_{j'}^{\ell} \cap
S_{j}^{k}$ and $x \not\in S_{\tilde{j}}^{k'}$ for any
$\tilde{j} \in \{1, \ldots s_{k'} \}$ and any $k'$ with $\ell < k' <
k$.

To make think clear let us observe the following example. Let $n =
7$. Assume,
\begin{itemize}
\item $S_1^3 = \{1,2,3\}, \;\; S_2^3 = \{4,5,6\}$
\item $S_1^2 = \{1,2\}, \;\; S_2^2 = \{3,4\}, \;\; S_3^2 = \{6,7\}$
\item $S_1^1 = \{3\},\;\; S_2^1 = \{4\},\;\;S_3^1 = \{5\}.$
\end{itemize}
The corresponding graph is given in Figure \ref{fig:dag}. Note that
$S_2^2$ has two parents, $S_2^3$ has parents from two different levels
and $S_3^1$ has parent in a which skips a level. Thus any such set of
constraints can be converted into a DAG.
\begin{figure}
\centering
\includegraphics[scale = 0.6]{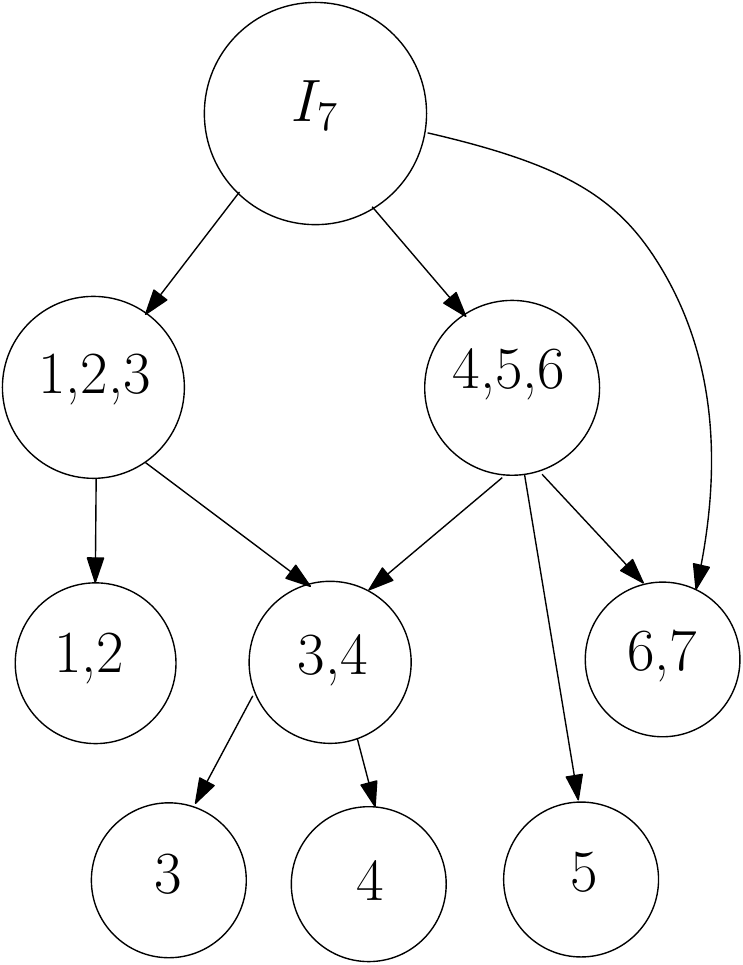}
\caption{\label{fig:dag}
The construction of a directed acyclic graph from the constraints of
an optimization problem.}
\end{figure}
\subsection{Splitting of offline and online problems}
Here, we discuss a theoretical justification for splitting the
optimization problem into an offline and online setting. The
procedure, though may have been initially thought about just as a
means to reduce computation time, we shall show here that it can have
a statistical justification. Moreover, the analysis will give us an
formal algorithm to split the optimization problem, instead of the
past heuristics.

Consider $G = (V,E)$ to be the directed acyclic graph created from the construction in Section \ref{sec:graph_const}. We assume that for each node $i$ there is a population mean $\mu_i$, and covariance structure $\Sigma_i$ for the parameter in the constraint. Moreover, we assume that we have taken $n_i$ independent samples $p_{ij}$ for $j = 1, \ldots, n_i$ from this population to create a constraint of the form $$ \sum_{j=1}^{n_i} x_{ij} p_{ij} \geq c_i.$$ Moreover, we assume that for a parent node, the population distribution is a mixture of its children distributions. Formally, the parent node has a mean and covariance given by the following result.

\begin{lemma}
\label{lem:mix_dist}
Let the parent node $p$ have $d$ child nodes each with mean and covariance $(\mub_i, \Sigmab_i)$ for $i = 1,\ldots, d$. Moreover assume that the distribution of $\pb$ is formed by mixing proportions $\alpha_i$ with $\sum_{i=1}^d \alpha_i = 1$. Then, the mean and covariance of the parent $(\mub_p, \Sigmab_p)$ is given by 
\begin{align*}
\mub_p &= \sum_{i=1}^d \alpha_i \mub_i\\
\Sigmab_p &= \sum_{i=1}^d \alpha_i \left( \Sigmab_i + (\mub_i - \mub_p)(\mub_i - \mub_p)^T\right).
\end{align*}
\end{lemma}

\begin{proof}
Let $Y_i$'s be independent samples drawn from node $i$ and $Z$ be a random variable which takes value $i$ with probablity $\alpha_i$ for $i = 1,\ldots, d$. Then $X = \sum_{i=1}^d Y_i \mathds{1}(Z = i)$ is a random variable from the parent distribution. Thus, we have
\begin{align*}
\mub_p &= \EE(X) = \EE(\EE(X|Z)) = \sum_{i=1}^d \alpha_i\mub_i
\end{align*}
and
\begin{align*}
\Sigmab_p &= \VV(X) = \VV(\EE(X|Z)) + \EE(\VV(X|Z)) \\
&= \VV\left(\sum_{i=1}^d \mathds{1}(Z = i) \mub_i \right) + \sum_{i=1}^d \alpha_i \Sigmab_i \\
&= \sum_{i=1}^d \alpha_i (\mub_i - \mub_p)(\mub_i - \mub_p)^T + \sum_{i=1}^d \alpha_i \Sigmab_i.
\end{align*}
Hence the result follows.
\end{proof}

Lemma \ref{lem:mix_dist} gives us a way to compute the variance at the parent node as a function of the child nodes. Using this result, we can compute the variance at each node $\Sigmab_i$. While the expected error accrued from using the dual estimate corresponding to node $i$ is monotonically related to $\Sigmab_i$, the actual function most likely does not have a closed analytic form. However, it can be estimated via sampling data. Let $\lambda (\frac{\Sigmab_i}{n_i}) $ denote the maximum eigenvalue of $\frac{\Sigmab_i}{n_i}$ which acts as a proxy for the error. Let $t(n_i)$ denote the time required to solve the QP under node $n_i$. Then, we would want to keep the dual estimate obtained for the node if both $t(n_i)$ is large and $\lambda (\frac{\Sigmab_i}{|n_i|})$ is small. We can use a convex combination of the two factors controlled by $w \in \sbr{0,1}$, and the resultant quantity is below a threshold $\beta$. This condition is expressed in Step \ref{step:threshold} of Algorithm \ref{algo:splitting}. 

\begin{algorithm}[!h]
\caption{Identify Stage 2 components}\label{algo:splitting}
\begin{algorithmic}[1]
\State Input : Directed acyclic graph $G = (V,E)$ formed by the optimization problem
\State Input: $w \in [0,1]$ to prioritize time required vs accuracy obtained in Stage 2
\State Input: $\beta$: threshold for stage 2
\State Output : Set of nodes with reliable dual estimates
\State Sample $N$ leaf nodes (users) to form $G_s$
\State $G_{out} = \emptyset$
\State Create a partial order $V_s^p$ of the nodes in $G_s$. Now traverse $V_s^p$ starting from the root
\For {$i = 1, \ldots, |V_s^p|$}
\State Compute $\Sigma_i$ using leaf nodes under node $n_i$
\If {$w \frac{1}{t(n_i)} + (1-w) \lambda (\frac{\Sigma_i}{|n_i|}) \leq \beta$} \label{step:threshold}
\State $G_{out} = G_{out} \cup n_i$
\EndIf
\EndFor
\State Return $G_{out}$
\end{algorithmic}
\end{algorithm}



\section{Practical challenges in scaling the solution}
\label{sec:scaling_QP}
In Sections \ref{sec:problem_def_global}, \ref{sec:problem_def_local}
and \ref{sec:graph_cal} we gave theoretical details regarding how to
handle user level constraints. Although the mathematical framework is
challenging and exciting, deploying such systems into production
raises a lot more issues. In this section, we outline two primary
challenges in this space and share our experiences in trying to
overcome them. The first challenge was to scale a QP solver to handle
a larger sample of data to get optimal dual variables -- this was
addressed by using operator splitting, an ADMM algorithm. The second
challenge was using variance reduction techniques on the sampled data
to reduce the variance of the dual estimates obtained from solving the
QP.

\subsection{Scaling up the QP Solver}
\label{sec:scaleup}

The dual problem for any optimization problem of the
form \eqref{eq:qp1} can be written as \eqref{eq:dual2}. To solve
the \eqref{eq:dual2} in an large-scale setting we employ the method of
operator splitting \cite{parikh2014block}. The generic problem for
operator splitting can be written as
\begin{equation*}
\begin{aligned}
& \underset{\zb}{\text{Minimize}} & & \phi(\zb)\\
& \text{subject to} & & \zb \in \mathcal{D}.\\
\end{aligned}
\end{equation*}
The operator splitting algorithm can be outlined as 
\begin{align*}
\zb^{k+1/2} &= \text{prox}_\phi(\zb^k  - \tilde{\zb}^k) \\ 
\zb^{k+1} &= \Pi_{\mathcal{D}}(\zb^{k+1/2} + \tilde{\zb}^k) \\
\tilde{\zb}^{k + 1} &= \tilde{\zb}^k + \zb^{k+1/2} - \zb^{k+1}, 
\end{align*}
where $\text{prox}_\phi(v) = \argmin_x\left( \phi(x) + \frac{\rho}{2}
\norm{ x - v}_2^2\right) $. In our case, the most expensive step is
evaluating the prox function. Let $\phi(\zb) = \frac{1}{2}\zb^T\Bb\zb
- \tilde{\pb}^T\zb$. Then it is easy to see that
\begin{align*}
\text{prox}_\phi(\vb) &= \argmin_\xb\left( \phi(\xb) + \frac{\rho}{2} \norm{ \xb - \vb}_2^2\right) \\
&= \argmin_\xb\left( \frac{1}{2}\xb^T\Bb\xb  - \tilde{\pb}^T\xb + \frac{\rho}{2} \norm{\xb - \vb}_2^2\right) \\
&= (\Bb + \rho \Ib)^{-1}(\tilde{\pb} + \rho \vb)
\end{align*}
Thus in our specific case \eqref{eq:dual2}, we write the algorithm as follows
\begin{align*}
\zb^{k+1/2} &= (\Bb + \rho \Ib)^{-1}(\tilde{\pb} + \rho(\zb^k  - \tilde{\zb}^k)) \\ 
\zb^{k+1} &=(\zb^{k+1/2} + \tilde{\zb}^k)_+ \\
\tilde{\zb}^{k + 1} &= \tilde{\zb}^k + \zb^{k+1/2} - \zb^{k+1}, 
\end{align*}
where $(\xb)_+ = \xb$ if $\xb \geq 0$ element-wise, else $0$. We run
these steps till convergence and return $\mu$ as the dual
variable. The true convergence is measured by the following
criteria. Define,
\begin{align*}
\rb^{k+1} = \zb^{k+1/2} - \zb^{k+1}, \qquad \sbb^{k+1} = -\rho(\zb^{k+1} - \zb^k).
\end{align*} 
where they can be regarded as the primal and dual residuals in the
algorithm. We stop the algorithm, when both the residuals are small,
i.e.,
\[
\norm{\rb^k}_2 \leq \epsilon^{pri} \;\; \text{and} \;\; \norm{\sbb^k}_2 \leq \epsilon^{dual}
\]
where the cutoff as obtained in \cite{boyd2011distributed} is given by
\begin{align*}
\epsilon^{pri} &= \sqrt{n}\epsilon^{abs} + \epsilon^{rel} \max\{ \norm{\zb^{k - 1/2}}_2, \norm{\zb^k}_2 \},\\
\epsilon^{dual} &= \sqrt{n}\epsilon^{abs} + \epsilon^{rel} \norm{ \rho \tilde{\zb}^k}_2.
\end{align*}
\begin{remark}
A smaller value of $\epsilon^{abs}$ and $\epsilon^{rel}$ will lead to
a more accurate solution but the algorithm will take much more time to
converge. To get an approximate solution, we can only check the
relative error in $\mu$. If it is below a certain threshold we stop
the algorithm. This is not exact, but we have seen that in most cases
it gives substantial improvement in convergence time.
\end{remark}
\begin{remark}
The scalability of the algorithm is dependent on finding the inverse
of the matrix $\Bb + \rho \Ib$. In most cases $\Bb$ is sparse, with
sparsity ratio $O(1/n)$ and with 64GB machine, we can find the sparse
Cholesky decomposition of $\Bb + \rho \Ib$ when $n \leq 5\times
10^7$. For larger problems, we would need a machine with more memory.
\end{remark}
\begin{remark}
We cannot use the Block Splitting algorithm \cite{parikh2014block} to
solve this problem because the matrix $\Bb$ is not separable. We could
have used it to solve the primal problem but in most cases we would
not be able to recover the dual from the primal solution.
\end{remark}
Using this technique we can solve the QP \eqref{eq:dual2} with
approximately $5 \times 10^6$ variables on a single machine with 64 GB
of RAM in a relatively small amount of time. This gives us a
$100 \times$ scale up from using the some of the off-the-solvers such
as \texttt{quadprog} in R. The convergence time for different problem
sizes is tabulated in Table \ref{tab:time}.
\vspace{-0.5em}
\begin{table}[!h]
\caption{\label{tab:time} Time taken for convergence for large scale Operator Splitting algorithm}
\begin{center}
 \begin{tabular}{||c c c c||} 
 \hline
 $n$ & Time per & Total number & Total time \\  
 &iteration & of iterations & \\[0.5ex]
 \hline\hline 
 $10^6$ & 0.165 &14135 & 38.87 (Minutes) \\ 
 \hline
 $10^7$ & 1.781 & NA  & NA \\
 \hline
 $10^8$ & 19.87 & NA  & NA \\
 \hline
\end{tabular}
\end{center}
\end{table}
For $n = 10^7$ and $n = 10^8$ each iteration takes about 1.7 and 20
seconds repectively. Thus, if we run the algorithm long enough to
converge, we estimate that the number of iterations should be about
$1.5\times 10^5$ and $1.5 \times 10^6$, resulting in a total time of
about $3$ days and a year respectively.

\subsection{Variance reduction of dual estimates}
The variation in the dual variables that we obtain are directly
related to the underlying $M$ variate distribution from which
$\pb_u, \rb_u$ is drawn for $u = 1,\ldots, N$. To get a sense of the
joint covariance structure from an actual data, we show two plots in
Figure \ref{fig:dist}. We consider a 6 variate distribution and the
first plot is the joint distribution of the top 3 most occurring
variables, i.e. they had most non-zero entries among all samples. The
second figure shows the joint distribution of the 3 least occurring
variables. We omit the axis labels to prevent disclosing sensitive information. 
\begin{figure*}
\centering
\includegraphics[width=\linewidth]{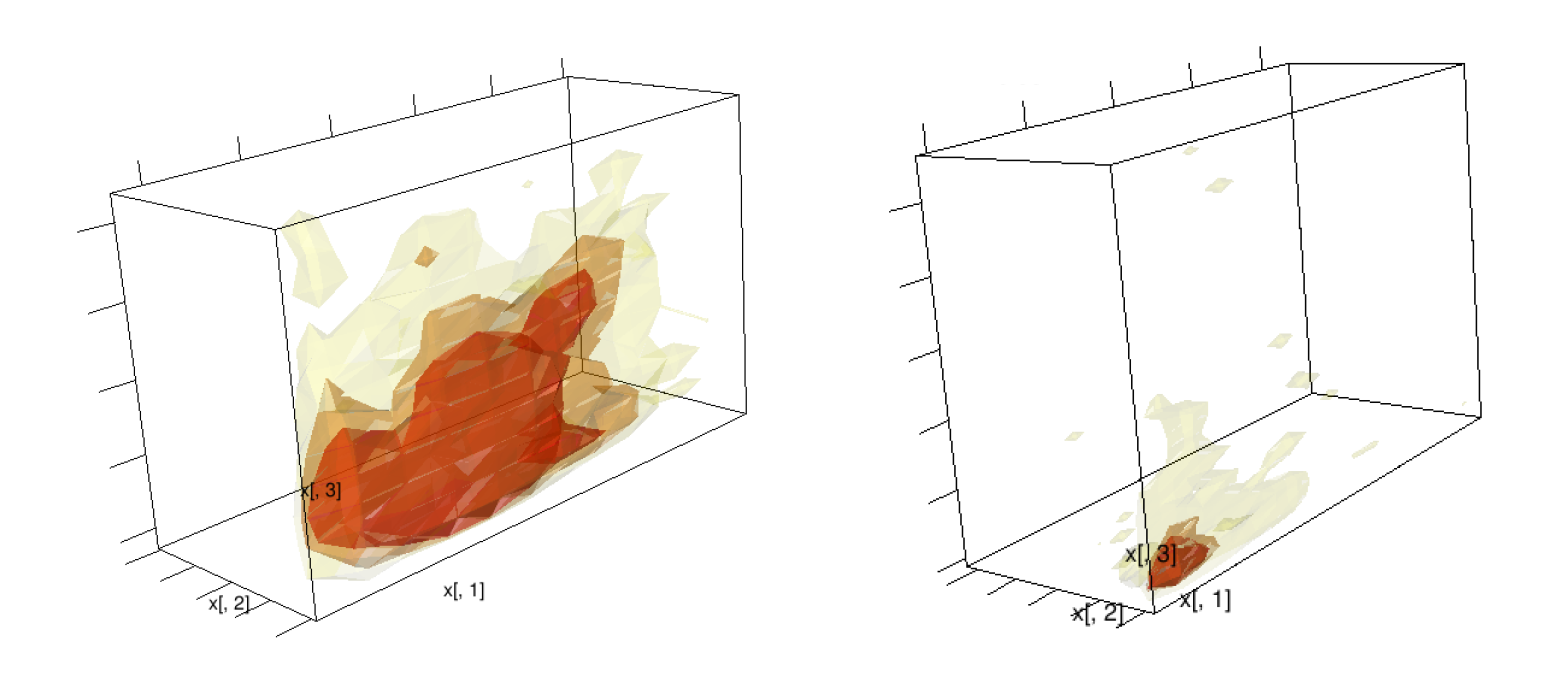}
\caption{\label{fig:dist}
The left and right panel shows the joint distribution of $p_i$
corresponding to the high and low density email keys respectively.}
\end{figure*}

Even though the non-sparse variables show a concentrated structure, we
see odd peaks for the more sparse variables. Thus, when we consider
the entire joint distribution the variation in $\mub$ is highly
dependent on the covariance structure of $\pb$ and $\rb$. From
Figure \ref{fig:dist}, it can be seen that the variance is quiet high
because of the odd peaks. This calls for certain variance reduction
techniques. Below we describe the three methods we have used and the
resulting effect in the estimate of $\mub$.
\subsubsection{Using linear moment matching}
If we try to solve for $\mub$ as a function of $\pb, \rb$, the general structure can be written as
\[
\mub = \Ab(\pb,\rb)^{-1}f(\pb,\rb)
\]
 where $\Ab(\pb, \rb)$ is a matrix and $f(\pb,\rb)$ is a vector both
 written as a function of the random variables $\pb$ and
 $\rb$. Studying the exact dependency of this function is extremely
 hard, and so we work on the heuristic that, if we can reduce the
 variation in $\pb, \rb$, we will reduce the variation in
 $\mub$. Towards that end, we use two control variates known as moment
 matching estimators. Since, we deal with $\pb$ and $\rb$ in the same
 way, we explain our technique only through $\pb$.
 
Assume the true mean and covariance of the distribution of $\pb$ is
denoted by $\mub_p$ and $\Sigmab_p$ both of which are unknown.  Using
the complete data, we estimate the mean of $\pb$, call it $\thetab
= \frac{1}{N}\sum_{u=1}^N \pb_u$ and the covariance matrix of $\pb$
call it $\Sigmab_{full} = \frac{1}{N}\sum_{u=1}^N(\pb_u
- \thetab)(\pb_u - \thetab)^T$. When we sample $n$ points $\pb_u$ from
this distribution for our problem, instead of using the sampled
$\pb_u$ in the optimization, we use two different modifications,
viz.  \begin{align} \label{mod1} \tilde{\pb}_u^1 &= \pb_u + \thetab
- \bar{\pb} &\forall \;\; u = 1,\ldots,
n\\ \label{mod2} \tilde{\pb}_u^2 &= \thetab
+ \Sigmab_{full}^{1/2}\Sigmab_{sampled}^{-1/2}\left(\pb_u
- \bar{\pb}\right) \qquad &\forall \;\; u = 1,\ldots, n \end{align}
where $\bar{\pb} = \frac{1}{n}\sum_{u=1}^n \pb_u$ and
$\Sigmab_{sampled} = \frac{1}{n}\sum_{u=1}^n(\pb_u - \bar{\pb})(\pb_i
- \bar{\pb})^T$ is the sample average and the sample covariance matrix
respectively.  Now we show few results for these two modified samples.
\begin{lemma}
For the modified sample $\tilde{\pb}_u^1$ given in \eqref{mod1} we have,
\begin{align*}
\mathbb{E}(\tilde{\pb}_u^1) &= \mub_p\\
\mathbb{V}(\tilde{\pb}_u^1) &\preceq \Sigmab_p
\end{align*}
where $\Ab \preceq \Bb$ implies $\Bb - \Ab$ is a positive semi-definite matrix.
\end{lemma}
\begin{proof}
The fact that $\tilde{\pb}_u^1$ is unbiased follows from the definition. To get the second assertion we evaluate the covariance $\tilde{\pb}_u^1$. 
\begin{align*}
\mathbb{V}(\tilde{\pb}_i^1) &= \mathbb{V}(\thetab) + \mathbb{V}(\pb_u) + \mathbb{V}(\bar{\pb}) + 2\text{Cov}(\thetab, \pb_u)\\
& -  2\text{Cov}(\pb_u, \bar{\pb}) - 2\text{Cov}(\thetab, \bar{\pb})\\
&= \frac{\Sigmab_p}{N} + \Sigmab_p + \frac{\Sigmab_p}{n} + 2\frac{\Sigmab_p}{N} - 2\frac{\Sigmab_p}{n} - 2\frac{\Sigmab_p}{N}  \\
&= \Sigmab_p \left( 1 + \frac{1}{N} - \frac{1}{n}\right)  \\
&\preceq \Sigmab_p
\end{align*}
where the last line follows from the fact that $n < N$.
\end{proof}
The second modification \eqref{mod2} does the moment matching for the
variance of $\pb$, but calculating the exact variance of this sample
is not possible in closed form. However, it can be shown that it is
asymptotically unbiased as $n \rightarrow N$, by applying the
Dominated Convergence Theorem.
\subsubsection{Using product moment matching}
Sometimes, we get much better results by using a product form of
moment matching. In this case we have,
\[
\tilde{\pb}_u^3 = \pb_i \times \frac{\thetab}{\bar{\pb}}
\]
where both the multiplication and division is done co-ordinate
wise. Boyle et. al. \cite{boyle1997monte} show that the moment
matching is asymptotically like using the known moments in control
variates. Further details regarding moment matching and control
variates can be found in \cite{mcbook}.

\section{Results}
\label{sec:results}
In this section, we present our experimental results. We first show
how the graph partitioning behaves as a function of time and accuracy.
We then show results concerning the three different moment matching
methods and how much accuracy we gain by solving the dual problem via
operator splitting.
\subsection{Graph Partitioning}
\begin{figure}[ht]
\centering
\includegraphics[width = \linewidth]{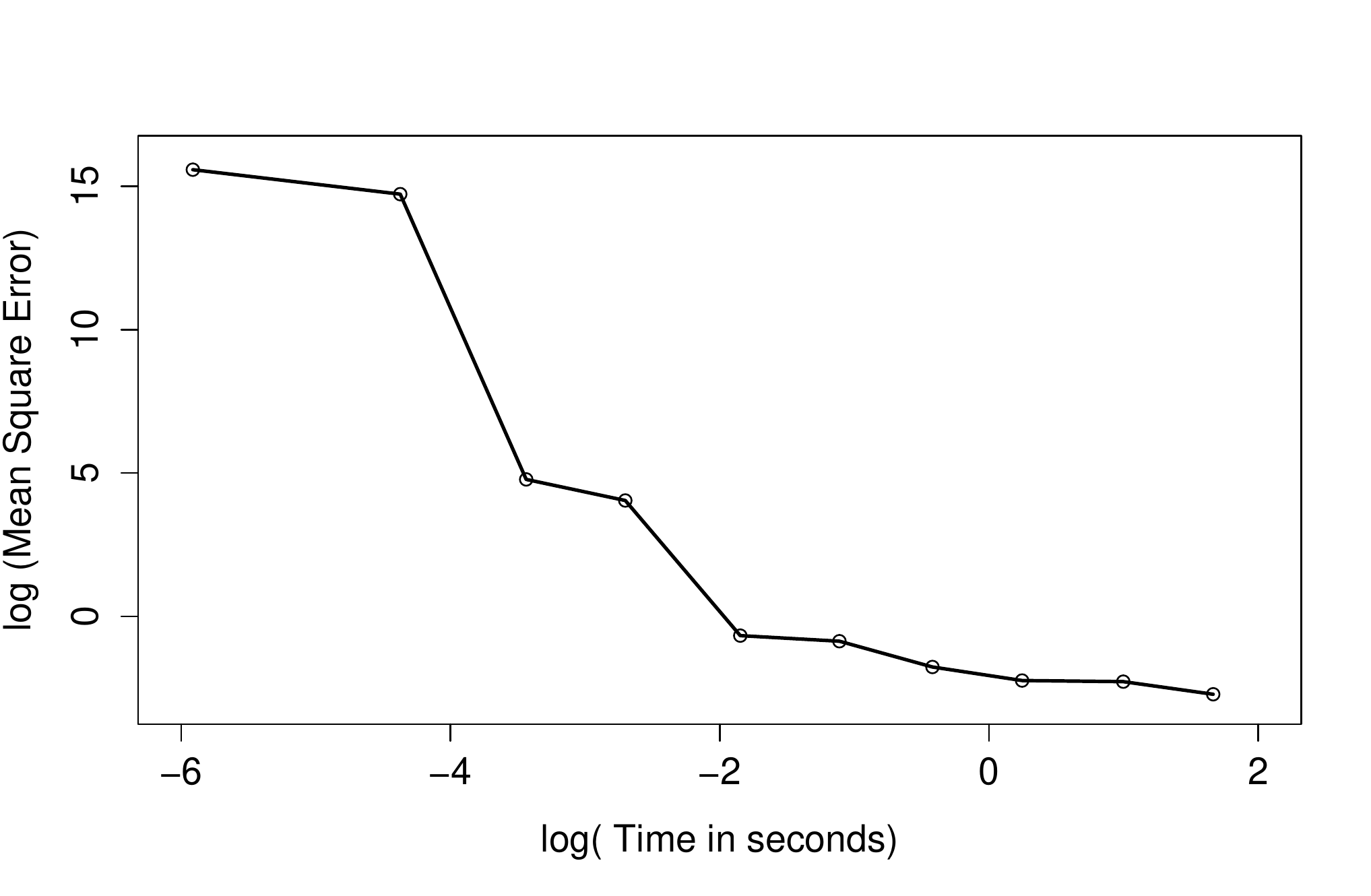}
\caption{\label{fig:time_err}The mean squared error versus computation
  time for a 10 level splitting procedure.}
\end{figure}
Consider the following toy example. Suppose we have $N = 1024$. Assume
that the constraints of the optimization problem can be written as a
binary tree having $K = 10$ levels. The root corresponds to a global
constraint. Every other node in the tree correspond to constraints on
half of the users from its parent node. The leaves are individual user
level constraints. Following the notation from Section
\ref{sec:graph_const}, we can write this as,
\begin{equation*}
\begin{aligned}
& \underset{\xb }{\text{Minimize}} & &f(\xb)\\
& \text{subject to} & & \sum_{i \in I_{1024}} x_i p_ i \geq c_1^n \\
 & \sum_{i \in S_j^{2^{10-k}}} x_i p_i \geq c_j^{2^{10-k}} & & \text{ for } j = 1,2,\ldots, 2^k \quad k = 1,\ldots,9\\
& & & x_i \geq c_i^{1} \text{ for } i = 1, \ldots, 1024.
\end{aligned}
\end{equation*}
Now we perform a split at level $k$ for $k = 1, \ldots, K$. When $k =
1$, we consider dual corresponding to only the global constraint and
solve two local problems having size 512. When $k = 2$, we consider
the duals corresponding to the first three constraints and then solve
4 local QPs having size 256. We carry on this procedure for $k = 1,
\ldots, K$. For each split, we obtain the mean square error (MSE) of
the final objective value by repeated experiments using random
samples. We also store the time taken to solve the online
problem. Figure \ref{fig:time_err} shows how the MSE decays as travel
up the tree.

Each point on the graph corresponds to a split in the tree. The
smallest computation time is when we cut the graph at $k = K$ and our
local projections are univariate problems. However, in this case, we
get least accurate solution since the only effect of the new sample is
in the final level projection resulting in low accuracy due to large
variance of $\pb$.
We also see that as we move up the tree, the accuracy increases along
with an increase in computation time. To get the optimal cut-off we
can use a weighted procedure such as Algorithm \ref{algo:splitting}.
\subsection{Large-scale Solver and Variance Reduction}
Consider the problem of sending emails pertaining to a single email
key. Suppose $r_{ui}$ is the probability that person $u$ will complain
if he is sent email $i$. Let $p_{ui}$ is the probability that there
will be a page view if email $i$ is sent and let $x_{ui}$ is the
probability that email $i$ is sent to user $u$. The problem that we
are interested in is basically to minimize the number of sends such
that the complains is reduced but page views do not suffer
much. Mathematically we can write this as,
\begin{equation*}
\begin{aligned}
& \underset{x}{\text{Minimize}} & & \xb^T\one + \frac{\gamma}{2}(\xb - \qb)^T(\xb-\qb)  \\
& \text{subject to} & & \sum_{u,i} x_{ui}r_{ui} \leq g_0 \\
&  &  &\sum_{u,i} x_{u,i}p_{u,i} \geq g_1 \\
& & & 0 \leq x_{u,i} \leq 1, \; u  = 1, \ldots, N; i = 1 ,\ldots, M.
\end{aligned}
\end{equation*}
Let $\mu_0$ and $\mu_1$ be the dual variables corresponding to the
complain and page views constraints which we are interested in. The
previous method is unable to solve this in a large scale setting. The
old infrastructure at Linkedin could solve this for about 10,000
variables to estimate $(\mu_0, \mu_1)$. Their variance estimate is
given in Table \ref{tab:var}. We use our Operator splitting algorithm
to solve the problem with a much larger data set ($\approx 10^6$
variables) on a single machine. We then perform the three variance
reduction techniques and the results are tabulated in Table
\ref{tab:var}.
\begin{table}[!h]
\caption{\label{tab:var}Comparing the Variance Reduction methods}
 \begin{center}
 \begin{tabular}{||c | c | c  | c | c||} 
 \hline
 Method & $\mu_0$ & $\mathbb{V}(\mu_0)$ & $\mu_1$ & $\mathbb{V}(\mu_1)$ \\ [0.5ex] 
 \hline\hline
 Old Method & $101.11$ & 16401.07 & 121.14 & 8.55 \\ 
 \hline
 Operator Splitting & $15.85$ & 40.04 & 120.41 & 0.42 \\
 \hline
 Using $\tilde{p}_i^1, \tilde{r}_i^1 $& 13.17 & 20.68 & 119.95 & 1.20 \\
 \hline
 Using $\tilde{p}_i^2, \tilde{r}_i^2 $& 12.96 & 12.48 & 120.21 & 0.99 \\
 \hline
 Using $\tilde{p}_i^3, \tilde{r}_i^3 $& 12.02 & 11.73 & 120.11 & 0.33 \\
 \hline
\end{tabular}
\end{center}
\end{table}
It can be observed from Table \ref{tab:var} that by just increasing
the sample size and solving a larger problem we have been able to
reduce the variance of the dual variable $\mu_0$ by about
$400\times$. The variance reduction techniques can further reduce the
variance by $\approx 4\times$. It can be seen that we get substantial
improvement in the dual estimates by following our procedure.

\vspace{1cm}
\section{Discussion}
Constrained optimization has proven to be a very successful formulation tool for several applications at Linkedin. Since most of our products aim to improve more than one metric, and (perhaps more importantly) impact several others (both positively and negatively), MOO approaches are now commonplace in various forms and flavours. Some of these applications include:

\begin{itemize}
\item Feed modeling: The Linkedin newsfeed is a very diverse distribution channel for various types of content. It surfaces articles shared by your network, job changes and anniversaries, profile updates, network updates (e.g., new connections made) among many other update types, in addition to serving advertisements. It is no surprise then that the feed ranking models have a significant impact on several objectives, including various forms of user engagement and revenue. The feed is also an application that is scored online (a user's visit triggers the scoring pipeline). The advances made in this paper can benefit the newsfeed ranking application by enabling user-level constraints like ``show no more than $3$ ads to a user in a day''. Scalability benefits are also applicable since the scale of the feed problem is in billions of data points, and both our scalability solutions can make MOO much more accurate and hence useful for the Linkedin feed.

\item Email and push notification portfolio optimization: Emails and push notifications are the most important (company-initiated) communication channel for any social network. Given the possible diversity of the channel, it is no surprise that the email and push notifications portfolio together drive a large plethora of metrics -- which makes it a perfect fit for MOO. Email portfolio optimization is a largely offline application (some parts can be user-action triggered and hence could be called borderline online). 
User level constraints are also very useful here, for instance to limit the number of emails that we send to any particular user should have a strict upper cap, so as not to overwhelm. The scale of the email problem is also into billions of data points, and hence can benefit greatly from both larger sample sizes and variance reduction.
\end{itemize}

Besides these two applications, there are many others where MOO is applicable and is either already being used (or should be used). After all, it is almost unthinkable that any application would care about just one objective. Instead of using heuristics or speculative weights, MOO is a much more principled and optimal approach to navigate the complex tradeoff.

Our contributions not only advance the theoretical state-of-the-art (by devising a way to exactly handle linear local constraints), but also make it much more scalable to obtain low variance dual estimates from $100 \times$ larger sample sizes. Together, these advances should make MOO solutions amenable for further wide-scale adoption.

\section*{Acknowledgments}
We would like to thank Prof. Art Owen for his helpful comments and support.

%
\bibliographystyle{abbrv}
\bibliography{qp}  
\end{document}